\documentclass[letterpaper, 10 pt, conference]{ieeeconf} 
\usepackage{graphicx}
\IEEEoverridecommandlockouts         
\usepackage{xcolor}
\usepackage{float}
\let\labelindent\relax
\usepackage{enumitem}
\setlist[enumerate]{leftmargin=*}
\newtheorem{assumption}{Assumption}
\newtheorem{theorem}{Theorem}
\newtheorem{remark}{Remark}
\usepackage{amsmath}
\usepackage{amssymb}
\usepackage{mathtools}
\usepackage{cancel}
\newcommand{\RR}{\mathbb{R}}
\def\J{\mathcal{J}}  
\def\L{\mathcal{L}}  
  
\def\S{\mathcal{S}} 
\def\P{\mathcal{P}} 
\def\V{\mathcal{V}} 

\title{\LARGE \bf
Wave-Based Bilateral Teleoperation between Nonlinear Manipulators with Direct Contact Force Feedback}

\author{G. Q. Bao Tran, Takanori Miyoshi, and Ho Duc Tho
\thanks{G.~Q.~B.~Tran is with Coordinated Science Laboratory, University of Illinois Urbana-Champaign, Urbana, IL 61801, USA (e-mail: {\tt\small baotran@illinois.edu}). His work is supported by the AFOSR MURI FA9550-23-1-0337 and NSF CMMI-2452534 grants. T.~Miyoshi is with Department of System Safety, Nagaoka University of Technology, 1603-1 Kamitomiokamachi, Niigata, Japan (e-mail: {\tt\small miyoshi@mech.nagaokaut.ac.jp}). H.~D.~Tho is with Department of Control Engineering and Automation, HCMC University of Technology (HCMUT), HCMC, Vietnam (e-mail: {\tt\small hoductho@hcmut.edu.vn}).}
}

\begin{document}

\maketitle
\thispagestyle{empty}
\pagestyle{empty}

\begin{abstract}
We study bilateral teleoperation between nonlinear, multi-DOF robotic manipulators in the presence of constant communication delays.
Unlike classical wave-transformation architectures that transmit a
coordinating force, we consider the case where
the environmental force is reflected to the master side to enhance teleoperation transparency. Since direct contact force feedback might destabilize the closed-loop system, we first develop a passivity-shortage characterization for the Euler--Lagrange remote system using a linear matrix inequality (LMI) approach. An upper strictly passive communication law is then employed to compensate for the computed passivity shortage so that the closed-loop stability under delays as well as position and force synchronization are preserved under appropriate conditions. Simulations with nonlinear $2$-DOF robotic manipulators in different settings illustrate our approach.
\end{abstract}

\section{INTRODUCTION}
\label{intro}

Bilateral teleoperation enables a human operator to interact with a remote environment through a master--slave robotic architecture. However, the presence of communication \emph{delays} may destabilize an otherwise stable closed-loop system, especially in force-reflecting teleoperation~\cite{yokokohji1999}. Passivity-based approaches have become a cornerstone for addressing delay-induced instability, starting from the scattering transformation framework in~\cite{anderson1988bilateral,spongtac} and its wave-transformation (WT) interpretation in~\cite{niemeyer}. Subsequent works refined this paradigm to handle time-varying delays, wave reflections, and performance limitations~\cite{lozano2002,tanner2005,chopra2006,aziminejad,hirche,nuno2011,polushin}. A well-known trade-off exists between stability and transparency in delayed teleoperation systems~\cite{lawrence1993}. In particular, when direct environmental contact forces are transmitted through a lossless passive WT channel, the remote robot may become non-passive, potentially destabilizing the closed-loop system. To avoid this issue, most WT-based architectures transmit a coordinating force rather than the true environmental contact force, which degrades transparency and telepresence quality.

In order to enable stable environmental contact force feedback to enhance transparency, a stabilizing communication law based on an \emph{upper strictly passive} (USP) communication channel is proposed for bilateral teleoperation systems with constant delays~\cite{tho2023tac}. 
The key idea is to deliberately introduce an \emph{excess level of passivity} in the communication channel to compensate for the passivity shortage of the remote subsystem when direct contact force feedback is used. By using Lyapunov and Lyapunov--Krasovskii techniques, delay-independent stability of the closed-loop system is established, and simulations illustrate improved transparency compared to classical lossless WT schemes. 
More recently, a force-receptive scattering transformation has been proposed in~\cite{tho2025neural}. 
By introducing additional communication ports and redundant scattering weights, the scheme allows kinematic and force information to be transmitted independently. 
These redundant parameters can be optimized numerically to improve teleoperation transparency while preserving passivity of the communication channel. 
However, both works rely on linear,
time-invariant robot models with a single degree-of-freedom (DOF). 
In many practical teleoperation applications, the robot dynamics are inherently nonlinear and multi-DOF, motivating the development of a framework that can handle nonlinear robot dynamics while retaining desirable communication and transparency properties.

This paper has two objectives. First, we extend the
stabilizing communication law framework of~\cite{tho2023tac} to bilateral
teleoperation systems with nonlinear multi-DOF Euler--Lagrange robots. Second, inspired by~\cite{tho2025neural}, we investigate
optimization-based tuning of the communication channel to
improve teleoperation transparency while preserving delay-robust stability.
Our main contributions are as follows:
\begin{itemize}[leftmargin=*,nosep]
\item A passivity-shortage characterization of the nonlinear remote subsystem via an LMI condition is developed;
\item A vector-valued multi-DOF extension of the USP communication law for nonlinear teleoperation guaranteeing stability as well as position and force synchronization under appropriate conditions;
\item An optimization-based tuning procedure is developed to improve teleoperation transparency.
\end{itemize}
Note again that the proposed approach is delay-independent, namely larger constant delays mainly degrade transparency and transient performance without affecting stability.

\emph{Notations:} Let $\RR^{m \times n}$ represent the set of real $m \times n$ matrices. 
Denote $\S_{>0}^n$ as the set of $n$-dimensional real positive definite symmetric matrices, and $I_n$ the $n$-dimensional identity matrix. Let $|\cdot|$ be the Euclidean vector norm. Let 
$\L_\infty$ (resp., $\L_2$) denote the space of uniformly bounded (resp., square integrable) functions~\cite[Page~196]{khalil}. In matrices, $\star$ denotes the entries determined by symmetry.

\section{PROBLEM DESCRIPTION}\label{sec:problem}

We consider a bilateral teleoperation architecture with constant communication delay $T > 0$, where a nonlinear local \emph{master} robot interacts with a human operator and a nonlinear remote \emph{slave} robot interacts with an environment (see Fig.~\ref{fig:tele}). 
\begin{figure}[!t]
    \centering
\includegraphics[width=0.49\textwidth]{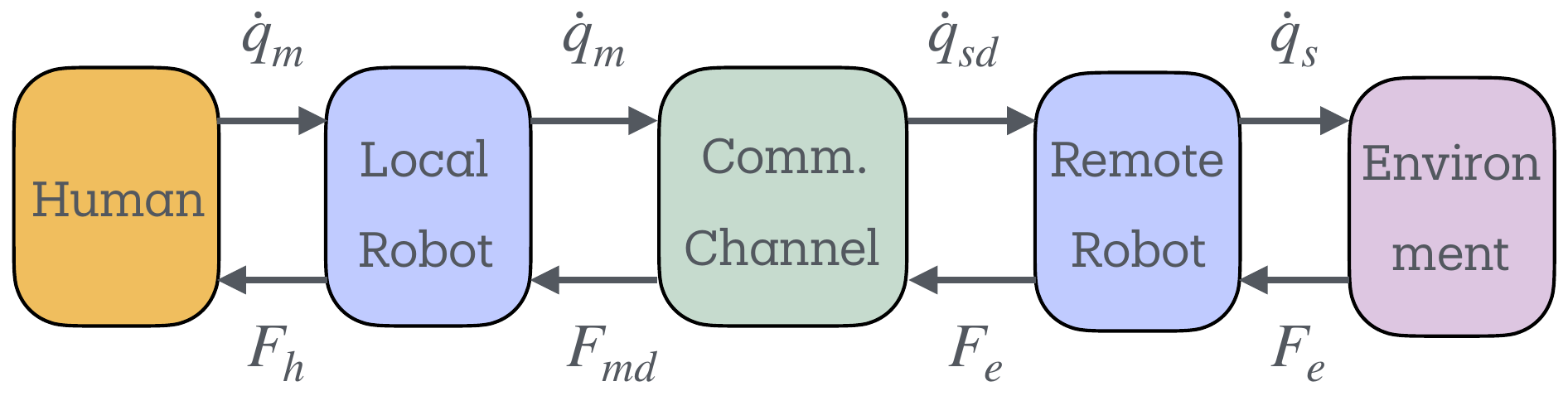}
    \caption{Velocity–force nonlinear teleoperation architecture with delayed communication channel.}
    \label{fig:tele}
\end{figure}
The two robots are modeled as Euler--Lagrange systems with generalized positions $q_m,q_s\in\RR^n$~\cite{nuno2008,chopra2008}
\begin{subequations}\label{eq:EL}
\begin{align}
M_m(q_m)\ddot{q}_m + C_m(q_m,\dot{q}_m)\dot{q}_m + G_m(q_m) & = \tau_m + F_h,\label{eq:master}\\
M_s(q_s)\ddot{q}_s + C_s(q_s,\dot{q}_s)\dot{q}_s + G_s(q_s) & = \tau_s - F_e,\label{eq:slave}
\end{align}
\end{subequations}
where $M_m(q_m),M_s(q_s)\in\S_{>0}^n$ are inertia matrices, $C_m(q_m,\dot q_m),C_s(q_s,\dot q_s) \in\RR^{n\times n}$ are Coriolis$\slash$centrifugal terms, $G_m(q_m),\tau_m,F_h\in\RR^n$ and $G_s(q_s),\tau_s,F_e\in\RR^n$ denote gravity, control, and interaction forces, respectively.
The matrices $M_i(q_i)$ and $C_i(q_i,\dot q_i)$ ($i \in \{m,s\}$) satisfy the standard properties of Euler--Lagrange systems~\cite{Murray}
\begin{itemize}[leftmargin=*,nosep]
\item $M_i(q_i)$ is symmetric positive definite for all $q_i$;
\item $\dot M_i(q_i)-2C_i(q_i,\dot q_i)$ is skew-symmetric.
\end{itemize}
In~\eqref{eq:master}, the human force $F_h$ exerted on the local robot is modeled as
\begin{equation}\label{eq:Fh}
F_h = -K_h q_m + f_h^\star,
\end{equation}
where $K_h \in \S^n_{>0}$ is the spring coefficient of the human hand. The exogenous force $f_h^\star$ generated by human muscles is
\begin{equation}\label{eq:fh}
f_h^\star = K_h q_{md} + F_h^\star,
\end{equation}
where $q_{md},F_h^\star \in \RR^n$ represent respectively the constant and time-varying part. We assume that $F_h^\star \in \L_\infty \cap \L_{2}$. Likewise, the environment force takes the form
\begin{equation}
F_e = K_e(q_s,\dot{q}_s) q_s,
\end{equation}
where $K_e(q_s,\dot{q}_s)\in\S_{>0}^n$ is the (state-dependent) stiffness matrix of the environment. In order to improve teleoperation transparency, the contact force $F_e$ is transmitted to the local side instead of a coordinating force. As seen in the linear, scalar setting~\cite{tho2023tac}, such a direct contact force feedback may render the remote subsystem (including the remote robot $+$ controller $+$ environment) non-passive and can in turn destabilize the closed-loop teleoperation under delays. In order to circumvent this major problem, a multi-input multi-output USP communication law will be established within the communication channel (see Fig.~\ref{fig:tele}) to dissipate excessive energy generated by the remote subsystem. 

The control inputs for the local and remote robots, employing both gravity compensation and damping injection, are given by\begin{subequations}
\begin{align}
\tau_m & = -B_m(q_m,\dot{q}_m)\dot{q}_m - F_{md} + G_m(q_m), \label{eq:controller_master}\\
\tau_s & = -B_{s1}(q_s,\dot{q}_s)\dot{q}_s + F_s + G_s(q_s), \label{eq:controller_slave}
\end{align}
\end{subequations}
where $B_m(q_m,\dot{q}_m)\in\S^n_{>0}$, $B_{s1}(q_s,\dot{q}_s)\in\RR^{n\times n}$ are (possibly) state-dependent damping matrices. Moreover, the coordinating force is
\begin{equation}
F_s = K_s(q_s,\dot{q}_s)\int_{0}^{t}(\dot{q}_{sd}-\dot{q}_s)d\tau + B_{s2}(q_s,\dot{q}_s)(\dot{q}_{sd}-\dot{q}_s),
\end{equation}
with $K_s(q_s,\dot{q}_s),B_{s2}(q_s,\dot{q}_s)\in\RR^{n\times n}$ the integral and proportional gains, respectively, where $(q_{sd},\dot{q}_{sd})$ is the remote command and $F_{md}$ is the reflected force at the master side, all produced by the communication law. Substituting~\eqref{eq:controller_slave} into~\eqref{eq:slave} under the assumption that $q_{sd}(0) = q_s(0)$, we get the closed-loop dynamics of the remote subsystem
\begin{subequations}\label{eq:slave_cl}
\begin{multline}
M_s(q_s)\ddot{q}_s + C(q_s,\dot{q}_s)\dot{q}_s + K(q_s,\dot{q}_s)q_s
=\\
K_s(q_s,\dot{q}_s)q_{sd} + B_{s2}(q_s,\dot{q}_s)\dot{q}_{sd},
\end{multline}
in which
\begin{align}
C(q_s,\dot{q}_s) & = C_s(q_s,\dot{q}_s)+B_{s1}(q_s,\dot{q}_s)+B_{s2}(q_s,\dot{q}_s),\\
K(q_s,\dot{q}_s) & = K_s(q_s,\dot{q}_s)+K_e(q_s,\dot{q}_s).
\end{align}
\end{subequations}
Similarly, by applying~\eqref{eq:controller_master} on~\eqref{eq:master}, we obtain the dynamics of the local subsystem
\begin{multline}\label{eq:localsys}
M_m(q_m)\ddot{q}_m + C_m(q_m,\dot{q}_m)\dot{q}_m + B_m(q_m,\dot{q}_m)\dot{q}_m \\= -F_{md} + F_h.
\end{multline}

In this paper, by using direct environmental force feedback, our objective is to design a communication channel so that the
closed-loop teleoperation system remains stable in the presence of
communication delays while achieving high-fidelity haptic realism. In other words, the main goals are to ensure that:
\begin{itemize}[leftmargin=*,nosep]
\item $q_s$ tracks $q_m$, i.e., the remote robot follows the local one during free motion (i.e., when not resisted by the environment);
\item $F_h$ tracks $F_e$, i.e., the human operator can feel the force exerted by the remote environment on the remote robot during hard contact. 
\end{itemize}
To facilitate the analysis, we assume \emph{a priori} boundedness of the slave motion over possibly large sets.
\begin{assumption}\label{ass:compact}
There exist compact sets $\P, \V \subset \RR^n$ such that $q_s(t) \in \P$ and $\dot{q}_s(t) \in \V$ for all $t \geq 0$. 
\end{assumption}

\begin{remark}
Assumption~\ref{ass:compact} requires the slave position and velocity to evolve within compact sets that can be large, while no convergence is assumed. In practice, robotic joints
operate within physical limits and the injected damping together with
the communication law prevents unbounded velocity growth. Consequently,
the slave trajectories remain confined to a bounded region of the state
space during normal operation. This assumption ensures that the matrices
$M_s(q_s)$, $C(q_s,\dot q_s)$, etc., remain bounded, which later
enables the LMI-based passivity characterization to be
verified over the set $\P\times\V$. Moreover, this assumption is imposed only on the slave (remote) robot, not on the master (local) robot. Since boundedness of all signals is established later, the analysis is not circular. Moreover, the boundedness of $q_s$ and $\dot q_s$ established later in Theorem~\ref{theo:stability} ensures that the trajectories remain inside some compact subset of $\RR^n\times\RR^n$, so the sets $\P$ and $\V$ can always be chosen to contain the reachable trajectories. Note that this first assumed boundedness does not guarantee stability (in the sense of arbitrarily bounding solutions by bounding their initial conditions) or asymptotic convergence, and we will later prove in Theorem~\ref{theo:stability} the convergence to $0$ of important signals, which is not assumed. In this spirit, the later results can be stated in a semi-global way, which will be investigated in our future work.
\end{remark}

Next, we characterize the passivity shortage of the remote subsystem, which will be used to design the USP stabilizing communication law. The following assumptions are made, which are standard in robotic systems~\cite[Pages~126--127]{lewis2004robot}.
\begin{assumption}\label{ass_robot}
The matrix $M_i$ ($i \in \{m,s\}$) is uniformly invertible as $M_i^{-1}$. The maps $M_i$, $M_i^{-1}$, $C_i$, $B_m$, $B_{s1}$, $B_{s2}$,
$K_s$, and $K_e$ are $C^1$. There exists $\underline b_m>0$ such that
\begin{equation}\label{eq:Bm}
B_m(p,v)\succeq \underline b_m I_n,
\qquad
\forall (p,v) \in \RR^n \times \RR^n.
\end{equation}
\end{assumption}

\section{PASSIVITY-SHORTAGE ANALYSIS OF THE NONLINEAR REMOTE SUBSYSTEM}
In order to quantify the amount of passivity shortage of the remote subsystem~\eqref{eq:slave_cl} when direct environmental contact force feedback is used, the input--output map $F_e$--$\dot{q}_{sd}$ will be analyzed. We first define the following augmented states
\begin{subequations}
\begin{align}
x &:=(\dot{q}_s,q_s,q_{sd})
\in\RR^{3n},\\
z &:=(\dot{q}_s,q_s,\dot{q}_{sd},q_{sd})
\in\RR^{4n}.
\end{align}
\end{subequations}
Let $P\in\S_{>0}^{3n}$ and consider the quadratic storage function
\begin{equation}
V_s(x) := x^\top P x.
\end{equation}

\begin{theorem}[Passivity shortage of slave]\label{theo:passivity_shortage}
Suppose that Assumptions~\ref{ass:compact}  and~\ref{ass_robot} hold. Then, along~\eqref{eq:slave_cl}, it holds that
\begin{equation}\label{eq:passive}
\dot{V}_s(x) \leq F_e^\top \dot{q}_{sd} + \alpha\dot{q}_{sd}^\top \dot{q}_{sd} - \lambda\dot{q}_s^\top \dot{q}_s,
\quad \forall t\geq 0,
\end{equation}
if there exist $\alpha>0$, $P\in\S_{>0}^{3n}$, and $\lambda>0$
such that for all $(p,v)\in\P \times \V$, the following LMI is satisfied
\begin{subequations}
\begin{equation}\label{eq:robust_LMI_clean}
A_2(p,v)^\top P A_1
+
A_1^\top P A_2(p,v)
+
A_3(p,v)
\preceq 0,
\end{equation}
where the matrices are defined as
\begin{align}
A_1 &=
\begin{pmatrix}
I_n & 0 & 0 & 0\\
0 & I_n & 0 & 0\\
0 & 0 & 0 & I_n
\end{pmatrix} \in \RR^{3n \times 4n},\\
A_2(p,v) &=
\begin{pmatrix}
A_{2a}(p,v) &
A_{2b}(p,v) &
A_{2c}(p,v) &
A_{2d}(p,v)\\
I_n & 0 & 0 & 0\\
0 & 0 & I_n & 0
\end{pmatrix}\notag\\
& \in \RR^{3n \times 4n},\\
A_3(p,v) &=
{\setlength{\arraycolsep}{4.4pt}%
\begin{pmatrix}
\lambda I_n & 0 & 0 & 0\\
\star & 0 & -\frac12 K_e(p,v) & 0\\
\star & \star & -\alpha I_n & 0\\
\star & \star & \star & 0
\end{pmatrix}}\in \RR^{4n \times 4n},
\end{align}
and the sub-matrices of $A_{2}(p,v)$ are given by
\begin{align}
A_{2a}(p,v) &= -(M_s(p))^{-1}C(p,v)\in \RR^{n \times n},\\
A_{2b}(p,v) & = -(M_s(p))^{-1}K(p,v)\in \RR^{n \times n},\\
A_{2c}(p,v) & =(M_s(p))^{-1}B_{s2}(p,v)\in \RR^{n \times n},\\
A_{2d}(p,v) & = (M_s(p))^{-1}K_s(p,v)\in \RR^{n \times n}.
\end{align}
\end{subequations}
\end{theorem}

\begin{proof}
Differentiating $V_s(x)=x^\top P x$ along trajectories gives
$\dot V_s(x)=\dot x^\top P x+x^\top P \dot x$. We also have $\dot{x} = (\ddot q_s,\dot q_s,\dot q_{sd})$. From~\eqref{eq:slave_cl} and Assumption~\ref{ass_robot}, we have
\begin{multline*}
\ddot q_s
=
-(M_s(q_s))^{-1}C(q_s,\dot q_s)\dot q_s
-(M_s(q_s))^{-1}K(q_s,\dot q_s)q_s
\\+(M_s(q_s))^{-1}B_{s2}(q_s,\dot q_s)\dot q_{sd}
+(M_s(q_s))^{-1}K_s(q_s,\dot q_s)q_{sd}.
\end{multline*}
Therefore, we get $\dot x =
A_2(q_s,\dot q_s)z$. Using this result and also $x=A_1z$, we deduce that
\begin{align*}
\dot V_s(x)
&=z^\top A_2(q_s,\dot q_s)^\top P A_1 z+
z^\top A_1^\top P A_2(q_s,\dot q_s) z\\
&=z^\top(A_2(q_s,\dot q_s)^\top P A_1+A_1^\top P A_2(q_s,\dot q_s))z.
\end{align*}
Since $K_e(q_s,\dot{q}_s)$ is symmetric, we obtain
\begin{align*}
F_e^\top \dot q_{sd}&=
q_s^\top K_e(q_s,\dot{q}_s)\dot q_{sd}
\\&=
\frac12
z^\top
\begin{pmatrix}
0 & 0 & 0 & 0\\
\star & 0 & K_e(q_s,\dot{q}_s) & 0\\
\star & \star & 0 & 0\\
\star & \star & \star & 0
\end{pmatrix}
z.
\end{align*}
Moreover, the following properties hold:
\begin{align*}
-\lambda \dot q_s^\top \dot q_s
&=
z^\top
\begin{pmatrix}
-\lambda I_n & 0 & 0 & 0\\
\star & 0 & 0 & 0\\
\star & \star & 0 & 0\\
\star & \star & \star & 0
\end{pmatrix}
z,\\
\alpha \dot q_{sd}^\top \dot q_{sd}&
=
z^\top
\begin{pmatrix}
0 & 0 & 0 & 0\\
\star & 0 & 0 & 0\\
\star & \star & \alpha I_n & 0\\
\star & \star & \star & 0
\end{pmatrix}
z.
\end{align*}
Collecting the above expressions yields
\begin{multline*}
\dot V_s(x)-(
F_e^\top \dot q_{sd}+
\alpha \dot q_{sd}^\top \dot q_{sd}-\lambda \dot q_s^\top \dot q_s)
\\=z^\top(
A_2(q_s,\dot q_s)^\top P A_1+A_1^\top P A_2(q_s,\dot q_s)+A_3(q_s,\dot q_s))z.
\end{multline*}
By Assumption~\ref{ass:compact}, condition~\eqref{eq:robust_LMI_clean} implies
$$
A_2(q_s,\dot q_s)^\top P A_1
+
A_1^\top P A_2(q_s,\dot q_s)
+
A_3(q_s,\dot q_s)
\preceq 0
$$
along trajectories. Hence,~\eqref{eq:passive} follows.
\end{proof}

\begin{remark}
According to Theorem~\ref{theo:passivity_shortage}, the remote subsystem is non-passive with a \emph{passivity shortage level} of $\alpha > 0$ and an \emph{excessive velocity dissipation} $-\lambda < 0$ when there exists a constant storage matrix $P$ solving~\eqref{eq:robust_LMI_clean}. To relax conservativeness especially in high-DOF robots, the analysis
could allow a state-dependent storage matrix
$P(p,v)$. In that case, the storage function for the remote subsystem becomes $V_s(x) = x^\top P(q_s,\dot q_s) x$, and a parameter-dependent LMI condition can be deduced, which is more tractable than~\eqref{eq:robust_LMI_clean} thanks to reduced conservativeness. In practice, the LMI~\eqref{eq:robust_LMI_clean} can be solved following either: i) polytopic method~\cite{apkarian95} (provided that we choose vertices so that~\eqref{eq:robust_LMI_clean} can be written as a convex combination of its vertex values, in exchange for higher conservativeness); or ii) grid-based method~\cite{wu95} (which is less conservative but requires a high enough number of grid points to be verified~\cite{sferlazza}).
\end{remark}

\begin{figure*}[ht]
    \centering    \includegraphics[width=0.6\linewidth]{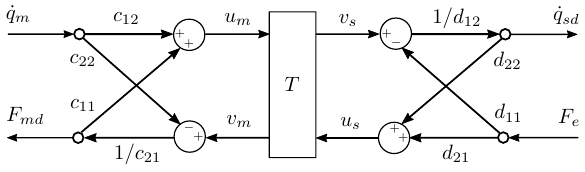}
    \caption{Communication channel governed by a USP wave transformation with vectorized inputs and outputs ($T$ is the constant time delay).}
    \label{fig:channel}
\end{figure*}

\section{WAVE-BASED STABILIZING COMMUNICATION LAW}

\subsection{Structure of Communication Law}
The dissipation inequality~\eqref{eq:passive} derived in the previous section shows that the remote subsystem exhibits a passivity shortage with respect to the port $\dot q_{sd}$ due to direct environmental contact force feedback. To guarantee stability of the delayed teleoperation loop, the communication channel must therefore supply an excess level of passivity that compensates for this shortage~\cite{tho2023tac}. In this section, we construct such a communication law by extending the scalar USP wave-transformation framework to vector-valued communication ports for nonlinear multi-DOF robotic manipulators.

The vector-valued USP communication structure is illustrated in Fig.~\ref{fig:channel}. Let $b > 0$ be the characteristic impedance. Denote $\gamma_l \leq 0$ and $\gamma_{r} \leq 0$ as the excessive level of passivity at the port $\dot{q}_{m}$ and $\dot{q}_{sd}$, respectively. The scattering weights $c_{ij}$ and $d_{ij}$ ($i,j \in \{1,2\}$) are given by~\cite{tho2023tac}
\begin{subequations}\label{eq:cd}
\begin{align}
 c_{11}&=b,&
c_{12}&=b\gamma_l+\frac{1}{4b},\label{eq:c1}\\
c_{21}&=b,&
c_{22}&=b\gamma_l-\frac{1}{4b}, \label{eq:c2}\\
d_{11}&=b,&
d_{12}&=-b\gamma_r+\frac{1}{4b},\label{eq:d1}\\
d_{21}&=b,&
d_{22}&=-b\gamma_r-\frac{1}{4b}.\label{eq:d2}
\end{align}
\end{subequations}
From Fig.~\ref{fig:channel}, the vector-valued wave variables $u_m$, $v_m$, $u_s$, $v_s$ $\in \RR^n$ are described as follows
\begin{subequations}
\begin{align}
u_m(t)& = c_{11}F_{md}(t)+c_{12}\dot{q}_m(t),\\
v_m(t) &= c_{21}F_{md}(t)+c_{22}\dot{q}_m(t),\\
u_s(t) &= d_{21}F_e(t)+d_{22}\dot q_{sd}(t),\\
v_s(t) &= d_{11}F_e(t)+d_{12}\dot q_{sd}(t),
\end{align}
\end{subequations}
together with the delay relations
\begin{subequations}
\begin{align}
v_s(t) &= u_m(t-T), \\
v_m(t) &= u_s(t-T).
\end{align}
\end{subequations}
By combining the wave definitions with the delay relations, we can express the channel
dynamics directly in terms of the physical variables
$(\dot q_m,F_e,\dot q_{sd},F_{md})$. Eliminating the wave variables
yields the following communication law
\begin{subequations}\label{eq:usp_law}
\begin{align}
\dot q_{sd}(t) &= \frac{4b^2\gamma_l+1}{-4b^2\gamma_r+1} \dot q_m(t-T) +\zeta_1(t),\\
F_{md}(t)& =F_e(t-T)+ \zeta_2(t),
\end{align}
where the residual terms $\zeta_{1}(t)$ and $\zeta_{2}(t)$ are
\begin{align}
\zeta_1(t)
&=
\frac{4b^2}{-4b^2\gamma_r+1}
(F_{md}(t-T)-F_e(t)),\\
\zeta_2(t)
&=
\frac{1}{b}
\bigg(
-\left(b\gamma_l-\frac{1}{4b}\right)\dot q_m(t)\nonumber\\&\qquad{} + \left(-b\gamma_r-\frac{1}{4b}\right)\dot q_{sd}(t-T)
\bigg).
\end{align}
\end{subequations}
From~\eqref{eq:usp_law}, the following results can be obtained
\begin{subequations}
\begin{equation}\label{eq:recursion}
\dot q_{sd}(t)
=
\delta \dot q_{sd}(t-2T)+\Delta(t),
\end{equation}
where
\begin{align}
\delta
&=
\frac{-4b^2\gamma_r-1}{-4b^2\gamma_r+1},\label{eq:delta}\\
\Delta(t)
&=
\frac{2}{-4b^2\gamma_r+1}
\dot q_m(t-T)
\notag\\&\qquad{}+
\frac{4b^2}{-4b^2\gamma_r+1}
(F_e(t-2T)-F_e(t)).\label{eq:Delta}
\end{align}
\end{subequations}
Thus, the delayed recursion governing $\dot q_{sd}$ has exactly the same form as in the scalar law~\cite{tho2023tac}.

\subsection{Stability and Transparency of Communication Law}

The stability and transmission properties of the USP communication law
follow from the delayed linear recursion satisfied by $\dot q_{sd}$. The coefficient $\delta$ depends only on $\gamma_r$ and satisfies
\begin{equation}
|\delta|<1
\qquad
\text{whenever } \gamma_r<0.
\end{equation}
Hence, the recursion~\eqref{eq:recursion} is exponentially stable. We make the following assumption regarding the initial histories.
\begin{assumption}\label{ass_history}
The delayed signals are initialized by bounded continuous histories
on $[-T,0]$. In particular, $u_m$, $u_s$, $v_m$, $v_s$, $\dot q_m$, $\dot q_{sd}$, $F_e$, and $F_{md}$
are supplied with bounded initial histories wherever
required by the delay relations.
\end{assumption}

The next theorem is the nonlinear, multi-DOF version of~\cite[Theorem~4]{tho2023tac}, which ensures boundedness of signals and especially the asymptotic convergence of velocities.

\begin{theorem}[Closed-loop stability]\label{theo:stability}
Suppose that Theorem~\ref{theo:passivity_shortage}
is satisfied with $\alpha>0$ and $\lambda>0$, and that Assumptions~\ref{ass_robot} and~\ref{ass_history} hold. If the communication parameter satisfies
\begin{equation}\label{eq:gamma}
\gamma_l \leq 0, \qquad \gamma_r < -\alpha,
\end{equation}
then all signals remain bounded. Moreover, as $t\to+\infty$,
\begin{equation}\label{eq:qmslim}
\dot q_m(t)\to 0, 
\qquad 
\dot q_s(t)\to 0.
\end{equation}
\end{theorem}

\begin{proof}
Consider the Lyapunov--Krasovskii functional
\begin{multline}\label{eq:V}
V
=
\frac12 \dot q_m^\top M_m(q_m)\dot q_m
\\+
\frac12 (q_{md} - q_m)^\top K_h (q_{md} - q_m)
\\+
V_s(x)
+
V_c,
\end{multline}
where $V_s(x)$ is the storage function from Theorem~\ref{theo:passivity_shortage} with
$x=(\dot q_s,q_s,q_{sd})$, and
\begin{equation}\label{eq:Vc}
V_c=\int_{t-T}^{t}(
|u_m(s)|^2+|u_s(s)|^2)ds
\end{equation}
is the communication-channel storage defined for the generalized wave transformation. Since $M_m(q_m)\in\S_{>0}^n$ for all $q_m \in \RR^n$, $K_h\in\S_{>0}^n$, $P\in\S_{>0}^{3n}$, and $V_c\geq 0$, we have
\begin{equation}\label{eq:V_pos}
V(t)\geq 0,
\qquad \forall t\geq 0.
\end{equation}
We differentiate $V$ in~\eqref{eq:V} along the closed-loop trajectories. First, we have
$$\frac{d}{dt}\left(\frac12 \dot q_m^\top M_m(q_m)\dot q_m\right)
=
\dot q_m^\top M_m(q_m)\ddot q_m
+\frac12 \dot q_m^\top \dot M_m(q_m)\dot q_m.$$ Second, since $K_h$ and $q_{md}$ are constant, we have $$\frac{d}{dt}\left(
\frac12 (q_{md}-q_m)^\top K_h (q_{md}-q_m)
\right)=
-(q_{md}-q_m)^\top K_h \dot q_m.$$
Hence, we get
\begin{multline}\label{eq:Vdot1}
\dot V
=
\dot q_m^\top M_m(q_m)\ddot q_m
+
\frac12 \dot q_m^\top \dot M_m(q_m)\dot q_m
\\-
(q_{md}-q_m)^\top K_h \dot q_m
+
\dot V_s(x)
+
\dot V_c.
\end{multline}
Next, by the property of Euler--Lagrange systems, $\dot M_m(q_m)-2C_m(q_m,\dot q_m)$ is skew-symmetric. Hence, $$\dot q_m^\top\left(\frac12 \dot M_m(q_m)-C_m(q_m,\dot q_m)\right)\dot q_m=0,$$
which implies that
\begin{multline*}
\dot q_m^\top M_m(q_m)\ddot q_m
+
\frac12 \dot q_m^\top \dot M_m(q_m)\dot q_m
\\=
\dot q_m^\top (M_m(q_m)\ddot q_m + C_m(q_m,\dot q_m)\dot q_m).
\end{multline*}
Substituting this into~\eqref{eq:Vdot1}, we obtain
\begin{multline*}
\dot V
=
\dot q_m^\top(M_m(q_m)\ddot q_m + C_m(q_m,\dot q_m)\dot q_m)
\\-
(q_{md}-q_m)^\top K_h \dot q_m
+
\dot V_s(x)
+
\dot V_c.
\end{multline*}
Using~\eqref{eq:localsys} and the symmetry of $K_h$, we get
\begin{multline}\label{eq:Vdot}
\dot V
=
-\dot q_m^\top B_m(q_m,\dot q_m)\dot q_m
-
\dot q_m^\top F_{md}
\\+
\dot q_m^\top F_h^\star
+
\dot V_s(x)
+
\dot V_c.
\end{multline}
Differentiating $V_c$ in~\eqref{eq:Vc} gives
\begin{align}
\dot V_c
&=
u_m^\top(t)u_m(t)
+
u_s^\top(t)u_s(t)
\notag\\&\qquad{}-
u_m^\top(t-T)u_m(t-T)
-
u_s^\top(t-T)u_s(t-T)\notag\\
& =
u_m^\top u_m-v_m^\top v_m
+
u_s^\top u_s-v_s^\top v_s.\label{eq:Vcdot}
\end{align}
For the master-side port, the wave-variable definitions yield
\begin{multline*}
u_m^\top u_m-v_m^\top v_m
=(c_{11}^2-c_{21}^2)F_{md}^\top F_{md}
\\
+
2(c_{11}c_{12}-c_{21}c_{22})
F_{md}^\top\dot q_m+
(c_{12}^2-c_{22}^2)
\dot q_m^\top\dot q_m.
\end{multline*}
Using~\eqref{eq:c1}-\eqref{eq:c2}, we get $u_m^\top u_m-v_m^\top v_m
=
F_{md}^\top\dot q_m
+
\gamma_l\dot q_m^\top\dot q_m$. Similarly, for the slave-side port,
\begin{multline*}
u_s^\top u_s-v_s^\top v_s
=
(d_{21}^2-d_{11}^2)F_e^\top F_e
\\+
2(d_{21}d_{22}-d_{11}d_{12})
F_e^\top\dot q_{sd}+
(d_{22}^2-d_{12}^2)
\dot q_{sd}^\top\dot q_{sd}.
\end{multline*}
Using~\eqref{eq:d1}-\eqref{eq:d2}, we get $u_s^\top u_s-v_s^\top v_s
=
-F_e^\top\dot q_{sd}
+
\gamma_r\dot q_{sd}^\top\dot q_{sd}$.
Plugging these into~\eqref{eq:Vcdot}, we obtain
\begin{equation}
\dot V_c
=
F_{md}^\top \dot q_m
-
F_e^\top \dot q_{sd}
+
\gamma_l \dot q_m^\top \dot q_m
+
\gamma_r \dot q_{sd}^\top \dot q_{sd}.
\label{eq:channel_dissipation}
\end{equation}
Substituting~\eqref{eq:passive} and
\eqref{eq:channel_dissipation} into~\eqref{eq:Vdot} gives
\begin{align*}
\dot V
{}&
\leq-\dot q_m^\top B_m(q_m,\dot q_m)\dot q_m
-
\cancel{F_{md}^\top \dot q_m}
+
\dot q_m^\top F_h^\star
\\
&\qquad{}
+
(\cancel{F_e^\top \dot q_{sd}}
+
\alpha \dot q_{sd}^\top \dot q_{sd}
-
\lambda \dot q_s^\top \dot q_s)\\
&\qquad{}
+
(\cancel{F_{md}^\top \dot q_m}
-
\cancel{F_e^\top \dot q_{sd}}
+
\gamma_l \dot q_m^\top \dot q_m
+
\gamma_r \dot q_{sd}^\top \dot q_{sd})\\
&=-\dot q_m^\top(B_m(q_m,\dot q_m)-\gamma_l I_n)\dot q_m
\\&\qquad{}
+
(\gamma_r+\alpha)\dot q_{sd}^\top \dot q_{sd}
-
\lambda \dot q_s^\top \dot q_s
+
\dot q_m^\top F_h^\star.
\end{align*}
Since $\gamma_l \leq 0$, define $\beta_m:=\underline b_m-\gamma_l>0$ with $\underline b_m$ from Assumption~\ref{ass_robot}. From~\eqref{eq:Bm}, we get
$$
B_m(p,v)-\gamma_l I_n\succeq \beta_m I_n, \qquad \forall (p,v) \in \RR^n\times\RR^n,
$$
hence
$$
\dot q_m^\top(B_m(q_m,\dot q_m)-\gamma_l I_n)\dot q_m
\geq
\beta_m |\dot q_m|^2.
$$
Therefore, we have
\begin{equation}\label{eq:Vdot2}
\dot V
\leq
-\beta_m |\dot q_m|^2
+
(\gamma_r+\alpha)|\dot q_{sd}|^2
-
\lambda |\dot q_s|^2
+
\dot q_m^\top F_h^\star.
\end{equation}
Using Young's inequality, we have
$$
\dot q_m^\top F_h^\star
\leq
\frac{\beta_m}{2}|\dot q_m|^2
+
\frac{1}{2\beta_m}|F_h^\star|^2.
$$
Substituting this into~\eqref{eq:Vdot2} yields
$$
\dot V
\leq
-\frac{\beta_m}{2}|\dot q_m|^2
+
(\gamma_r+\alpha)|\dot q_{sd}|^2
-
\lambda |\dot q_s|^2
+
\frac{1}{2\beta_m}|F_h^\star|^2.
$$
Let $-(\gamma_r+\alpha)=\eta>0$. We get
\begin{equation}\label{eq:Vdot3}
\dot V
\leq
-\frac{\beta_m}{2}|\dot q_m|^2
-
\eta |\dot q_{sd}|^2
-
\lambda |\dot q_s|^2
+
\frac{1}{2\beta_m}|F_h^\star|^2.
\end{equation}
We now integrate both sides of~\eqref{eq:Vdot3} from $0$ to $t$ and get
\begin{multline}
V(t)-V(0)
\leq
-\frac{\beta_m}{2}\int_0^t |\dot q_m(s)|^2 ds
-
\eta \int_0^t |\dot q_{sd}(s)|^2 ds
\\-
\lambda \int_0^t |\dot q_s(s)|^2 ds
+
\frac{1}{2\beta_m}\int_0^t |F_h^\star(s)|^2 ds.
\end{multline}
Thanks to~\eqref{eq:V_pos}, rearranging terms gives
\begin{multline}\label{eq:V_int}
\frac{\beta_m}{2}\int_0^t |\dot q_m(s)|^2 ds
+
\eta \int_0^t |\dot q_{sd}(s)|^2 ds
+
\lambda \int_0^t |\dot q_s(s)|^2 ds
\\
\leq
V(0)
+
\frac{1}{2\beta_m}\int_0^t |F_h^\star(s)|^2 ds.
\end{multline}
Because $F_h^\star \in \L_2$, the right-hand side of~\eqref{eq:V_int} is bounded uniformly in $t$. Letting $t\to+\infty$, we conclude that $\dot q_m,\dot q_{sd},\dot q_s\in \L_2$.
Next, from the fact that $V(t)\geq 0$, the integrated inequality also shows that $V \in \L_\infty$.

We next establish boundedness of the state variables from the boundedness of $V$. Since $K_h\in\S^n_{>0}$ is constant, boundedness of $V$ implies boundedness of $q_{md}-q_m$, and since $q_{md}$ is constant, this implies boundedness of $q_m$. Because $M_m(q_m)\in\S^n_{>0}$ is continuous and $q_m$ is bounded, boundedness of $V$ implies boundedness of $\dot q_m$. Similarly since $P\in\S_{>0}^{3n}$, boundedness of $V_s(x)$ implies boundedness of $x=(\dot q_s,q_s,q_{sd})$. Thus,
$q_m,\dot q_m,q_s,\dot q_s,q_{sd}$ are bounded. 

Since $F_h=K_h(q_{md}-q_m)+F_h^\star$, and since $q_m$ is bounded and $F_h^\star\in \L_\infty$ by assumption, it follows that $F_h\in \L_\infty$. Because $(q_s,\dot q_s)\in\P\times\V$, $K_e$ is bounded on
$\P\times\V$, and $q_s$ is bounded, it follows from
$F_e=K_e(q_s,\dot q_s)q_s$ that $F_e\in\L_\infty$. Using~\eqref{eq:recursion}, we analyze the behavior of $\dot q_{sd}$. Since $\dot q_m\in \L_\infty$ and $F_e\in \L_\infty$, it follows from~\eqref{eq:Delta} that $\Delta\in \L_\infty$.
Hence, from the recursion~\eqref{eq:recursion} with $|\delta| < 1$, we obtain
$\dot q_{sd}\in \L_\infty$.
Using~\eqref{eq:usp_law}, we then get $F_{md}\in \L_\infty$. We thus get $\tau_m \in \L_\infty$ and $\tau_s \in \L_\infty$, and in turn $\ddot{q}_m \in \L_\infty$ and $\ddot{q}_s \in \L_\infty$. By differentiating~\eqref{eq:recursion} and noting that $\ddot{q}_{m}, \ddot{q}_{s}, \dot{q}_{s} \in \L_\infty$, we can infer that $\ddot{q}_{sd} \in \L_\infty$. Because $\ddot q_m,\ddot q_s \in \L_\infty$ and
$\dot q_m,\dot q_s \in \L_2$,
Barbalat-type arguments (see~\cite{farkas}) imply~\eqref{eq:qmslim}.
\end{proof}

\begin{remark}
The use of $K_h$ in the functional~\eqref{eq:V} prevents us from considering a state-dependent stiffness, i.e., $K_h(q_m,\dot{q}_m)$, in~\eqref{eq:Fh}. Indeed, if this were the case, the derivative $\dot{K}_h$ would show up in $\dot{V}$ as some extra terms that are not easy to handle. But note that our results provide sufficient but not necessary conditions.
\end{remark}

Now that stability is ensured, the next theorem, which is the nonlinear, multi-DOF version of~\cite[Theorem~5]{tho2023tac}, guarantees teleoperation transparency under appropriate conditions. Denote the position and force tracking errors\begin{subequations}
\begin{align}
e_q(t) &:= q_s(t)-q_m(t-T),\\
e_f(t) &:= F_h(t)-F_e(t-T).
\end{align}
\end{subequations}
\begin{theorem}[Teleoperation transparency]\label{theo:transmission}
Suppose that Theorem~\ref{theo:passivity_shortage}
is satisfied and the channel is designed as described in Theorem~\ref{theo:stability}. Then, the following statements hold.

\begin{enumerate}[leftmargin=*,nosep]
\item \emph{(Free motion):}
Assume that $K_s(p,0)$ is non-singular for all $p\in\P$. If $\gamma_r=-\frac{1}{4b^2}$ and $F_e\equiv 0$, then
\begin{equation}\label{eq:limeq}
\lim_{t\to+\infty} e_q(t)=q_s(0)-q_m(-T).
\end{equation}
Thus, if initial positions satisfy
$q_s(0)=q_m(-T)$,\footnote{As per Assumption~\ref{ass_history}, initial histories are defined on $[-T,0]$ (negative time). While the condition $q_s(0)=q_m(-T)$ is rather restrictive, it is only needed for exact asymptotic synchronization. Otherwise, a constant bias remains as in~\eqref{eq:limeq}.} then
\begin{equation}
e_q(t)\to 0\text{ as }t\to+\infty;
\end{equation}
\item \emph{(Hard contact):} Assume that $K_e$ is $C^1$ on $\P \times \V$. If $\ddot q_m(t)\to 0$ as $t\to+\infty$,
then
\begin{equation}
e_f(t)\to 0 \text{ as }t\to+\infty .
\end{equation}
\end{enumerate}
\end{theorem}

\begin{proof}
First, we prove the free-motion result.
Substituting $\gamma_r=-\frac{1}{4b^2}$ and $F_e\equiv0$
into~\eqref{eq:recursion} yields $\dot q_{sd}(t)=\dot q_m(t-T)$. Integrating and rearranging this give
\begin{equation}\label{eq:pos_relation}
q_{sd}(t)-q_m(t-T)=q_{sd}(0)-q_m(-T).
\end{equation}
Next, we differentiate~\eqref{eq:slave_cl}.
Since $\ddot q_{s},\dot q_{s},\ddot q_{sd}$ are bounded 
(from Theorem~\ref{theo:stability}),
the jerk $\dddot q_s$ is bounded.
Hence, $\ddot q_s$ is uniformly continuous.
Because $\dot q_s(t)\to0$ as $t\to+\infty$
(from Theorem~\ref{theo:stability}),
Barbalat's lemma~\cite[Page~323]{khalil} implies $\ddot q_s(t)\to 0$. Taking $t\to+\infty$ in the free-motion slave dynamics yields
$$
K_s(q_s(+\infty),0)(q_s(+\infty)-q_{sd}(+\infty))=0.
$$
Because $K_s(p,0)$ is non-singular for all $p\in\P$, we obtain
\begin{equation}\label{eq:qsd_qs_zero}
q_{sd}(t)-q_s(t)\to 0.
\end{equation}
Combining~\eqref{eq:pos_relation} and~\eqref{eq:qsd_qs_zero}
gives
$$
\lim_{t\to+\infty}(q_s(t)-q_m(t-T))
=
q_s(0)-q_m(-T).
$$

We now prove the hard-contact result.
Taking the limit $t\to+\infty$ in~\eqref{eq:localsys} and using
$\dot q_m(t)\to0$ (from Theorem~\ref{theo:stability})
together with the assumption $\ddot q_m(t)\to0$
give
\begin{equation}\label{eq:Fh_Fmd}
\lim_{t\to+\infty}(F_h(t)-F_{md}(t))=0 .
\end{equation}
Next, we analyze the difference $F_e(t)-F_e(t-2T)$.
Since $F_e=K_e q_s$ and $\dot q_s\in\L_\infty$
with $\dot q_s(t)\to0$, we obtain using the Lebesgue dominated convergence theorem~\cite{sherbert1992}
\begin{align*}
&\lim_{t\to+\infty}(F_e(t)-F_e(t-2T))
\\&=\lim_{t\to+\infty}\int_{t-2T}^{t}\dot F_e(s)ds
\\
&=\int_{-2T}^{0}
\lim_{t\to+\infty}\bigg(
K_e(q_s(t+s),\dot q_s(t+s))\dot q_s(t+s)
\\&\qquad{}+
\dot K_e(q_s(t+s),\dot q_s(t+s))q_s(t+s)
\bigg)ds,
\end{align*}
where $\dot K_e(q_s,\dot q_s)$ is given by 
$$
\dot K_e(q_s,\dot q_s)
=
\sum_{i=1}^n \frac{\partial K_e}{\partial q_{s,i}}(q_s,\dot q_s)\dot q_{s,i}+
\sum_{i=1}^n \frac{\partial K_e}{\partial \dot q_{s,i}}(q_s,\dot q_s)\ddot q_{s,i},
$$
where $\dot{q}_{s,i}$ (resp., $\ddot{q}_{s,i}$) is the $i$-th component of $\dot{q}_s$ (resp., $\ddot{q}_s$). Because $K_e$ is $C^1$ on $\P \times \V$ compact, $K_e$ and its partial derivatives are bounded on $\P \times \V$, and since $q_s$ is bounded, while $\dot q_s(t)\to 0$ and $\ddot q_s(t)\to 0$, we obtain
\begin{equation}
\label{eq:Fe_FetmT_0}
\lim_{t\to+\infty}(F_e(t)-F_e(t-2T))
= 0.
\end{equation}
Applying~\eqref{eq:Fe_FetmT_0} together with $\dot q_m(t)\to 0$
to~\eqref{eq:Delta}, we get that $\Delta(t) \to 0$. Hence,~\eqref{eq:recursion} becomes a stable recursive equation with a vanishing disturbance, implying $\dot{q}_{sd}\to0$. Therefore, from~\eqref{eq:usp_law}, we have that
$$
\lim_{t\to+\infty}(F_{md}(t)-F_e(t-T))=0.
$$
Finally, combining this with~\eqref{eq:Fh_Fmd} yields the second claim.
\end{proof}

Theorem~\ref{theo:transmission} therefore shows that the proposed communication law achieves asymptotic position transmission in free motion and force transmission during hard contact. Note that condition $\gamma_r=-\frac{1}{4b^2}$ is not imposed in the hard-contact case.

\subsection{Optimization-Based Tuning of Communication Law}

The USP communication law depends on the
scattering parameters $b$ and $\gamma_l$ (note that for impedance matching we impose $\gamma_r = -\frac{1}{4b^2} < 0$). While the
condition $\gamma_r<0$, which determines the excess passivity injected by the communication channel, guarantees stability of the delayed recursion~\eqref{eq:recursion}, the particular choice of $b$ and
$\gamma_l$ can strongly affect transparency and
performance. We therefore view the communication law as a parameterized mapping
with tunable parameters
\begin{equation}
\theta := (b,\gamma_l).
\end{equation}
For each choice of $\theta$, the communication law~\eqref{eq:usp_law} generates the signals $\dot q_{sd}$ and $F_{md}$ that determine the interaction forces experienced by the human
operator.

To improve transparency, the parameter $\theta$ can be optimized offline so as to improve both position and force
tracking between the master and slave sides. We therefore consider
the combined weighted objective
\begin{equation}\label{eq:cost}
\J(\theta)
=
\int_0^{T_w}
(w_q |e_q(t)|^2 + w_f |e_f(t)|^2) dt,
\end{equation}
over a time window $T_w > 0$ larger than the transient phase, where $w_q \geq 0$ and $w_f \geq 0$ are normalized weighting coefficients that determine
the relative importance of position and force
transparency. The optimization problem is then, with $\alpha$ from Theorem~\ref{theo:passivity_shortage},
\begin{subequations}
\begin{align}
\min_{\theta} \quad
& \J(\theta) \\
\text{s.t.} \quad
& 0 < b < \frac{1}{2\sqrt{\alpha}}, \quad \gamma_l \leq 0.\label{eq:cons2}
\end{align}
\end{subequations}
Then, with the obtained $b$, we take $\gamma_r = -\frac{1}{4b^2}$. Note that the constraints~\eqref{eq:cons2} already impose that $\gamma_r < -\alpha$ to guarantee stability according to Theorem~\ref{theo:stability}. The
parameter $\gamma_l$ mainly affects the transmission gain from the
master velocity to the slave command and can therefore be tuned to
improve transparency. The optimization can be carried out offline using simulation-based evaluation
of $\J(\theta)$, yielding communication parameters that preserve
stability while improving transmission. While the optimization uses finite-horizon simulations, stability and transparency are guaranteed beyond the time $T_w$ as long as the parameters found satisfy Theorems~\ref{theo:stability} and~\ref{theo:transmission}.

\section{NUMERICAL ILLUSTRATION}
To illustrate the proposed framework, numerical simulations are conducted on a nonlinear two-DOF planar manipulator for both the master and slave robots.

\subsection{Robot Dynamics and Simulation Scenario}

Both the master and slave robots are modeled as the same nonlinear Euler--Lagrange system after gravity compensation\begin{subequations}
\begin{equation}\label{eq:robot_eg}
M(q)\ddot{q} + C(q,\dot{q})\dot{q} + B \dot{q} = \tau + F,
\end{equation}
where $q = (q_1,q_2)\in \RR^2$ denotes joint positions and $F$ denotes the external interaction force (human force for the master and environment force for the slave). The inertia matrix in~\eqref{eq:robot_eg} is given by~\cite[Page~165]{Murray}
\begin{equation}
M(q) =
\begin{pmatrix}
\frac{17}{4} + 2\cos(q_2) & 2 + \cos(q_2) \\
2 + \cos(q_2) & 2
\end{pmatrix},
\end{equation}
in kg, and the Coriolis matrix is
\begin{equation}
C(q,\dot{q}) =
\begin{pmatrix}
-\sin(q_2)\dot{q}_2 & -\sin(q_2)(\dot{q}_1 + \dot{q}_2) \\
\sin(q_2)\dot{q}_1 & 0
\end{pmatrix}.
\end{equation}
\end{subequations}
We take $K_h = 20I_2$~(N$\slash$m) and $F_h^\star = (0,0)$~(N). The viscous damping matrix is $B_m = B_{s1} = 0.5 I_2$~(Ns$\slash$m). The coordinating force gains are $K_s = 100 I_2$ and $B_{s2} = 20 I_2$. This choice of $B_{s2}$ corresponds to the parameter selection used in the theoretical development of the USP communication law. Here, the robots are nonlinear and have two DOFs while satisfying standard
Euler--Lagrange properties.

To simulate both the free-motion and hard-contact cases, we assume that there is a wall with a stiffness of $K_e = 100 I_2$ (which is moderately hard) placed at the position $q_s = 0$ in the negative side (component-wise). As a result, contact is made whenever $q_s$ reaches $0$ from a positive value, and the actual environment force can be described by
\begin{equation}
F_e(q_s) = K_e \min(q_s,0),
\end{equation}
where the minimum is taken component-wise. This corresponds to a unilateral contact model in which the slave robot interacts with the environment only when $q_s$ penetrates the surface, while no force is applied during free motion. To alternate between these two cases, we let $q_{md}$ be a periodic step function switching between $q_{md} = \pm 0.1$~(m) with a period of $60$~(s), spending half of the period in each mode, which is piece-wise constant. Note that this switching $F_e$ is only for the sake of illustrative simulations and we still treat $F_e = K_e q_s$ in the channel design. We introduce a constant delay $T = 0.2$~(s) in both transmission directions.

\subsection{Passivity Analysis and Communication Law Design}

To evaluate the passivity-shortage condition, with the matrices depending on $q_2$ but not $q_1$ and under assumptions about the boundedness of $q_s,\dot{q}_s$, we solve~\eqref{eq:robust_LMI_clean} on the set $\P \times \V = [-\pi,\pi]^2 \times [-1,1]^2$ following a grid-based approach~\cite{wu95} with step size $0.1$, using YALMIP~\cite{Lofberg2004}. We impose $\lambda = 0.001$ and obtain $\alpha = 5.7709$ together with a storage matrix $P$. Although this approach only verifies~\eqref{eq:robust_LMI_clean} at the grid points and not necessarily over the entire set $\P \times \V$, it provides a simple approximate certificate for dissipation.

To emphasize force transmission, the optimization minimizes the cost~\eqref{eq:cost} with $w_q = 0$ and $w_f = 1$, which is evaluated through $120$-second long simulations of the closed-loop teleoperation system at each iteration. 
The optimization problem is solved using MATLAB's \texttt{patternsearch} 
solver, where the variables are constrained in $0.03 \leq b \leq 0.09$ and 
$-25 \leq \gamma_l \leq -15$. The optimized parameters that we can find are $b = 0.06$ and $\gamma_l = -20$, leading to $\gamma_r = -69.4444$. Notice that since the cost $\J(\theta)$ is evaluated through time-domain simulations, it is generally non-convex and non-smooth; therefore, gradient-based optimization methods such as \texttt{fmincon} may exhibit poor performance or convergence issues. Instead, the derivative-free \texttt{patternsearch} iteratively explores a set of directions on a mesh around the current iterate and accepts steps that decrease the objective function, without requiring gradient information.

\subsection{Simulation Results}
Fig.~\ref{fig:position} shows the resulting closed-loop position trajectories, which remain stable, and the position transmission errors, which can only converge to constant offsets in the hard-contact case. Indeed, the environment's stiffness prevents the slave's position from penetrating as far as the master's. The magnitude of the slave velocities $\dot{q}_s$ is smaller than~$0.06$~(m$\slash$s), well below the value assumed when solving LMIs. Notice that $q_s$ in the simulations changes signs, so we are indeed alternating between free-motion and hard-contact scenarios, with contact made at around $33$~(s) and each $60$~(s) after that.

\begin{figure}[!t]
    \centering
\includegraphics[width=0.49\linewidth]{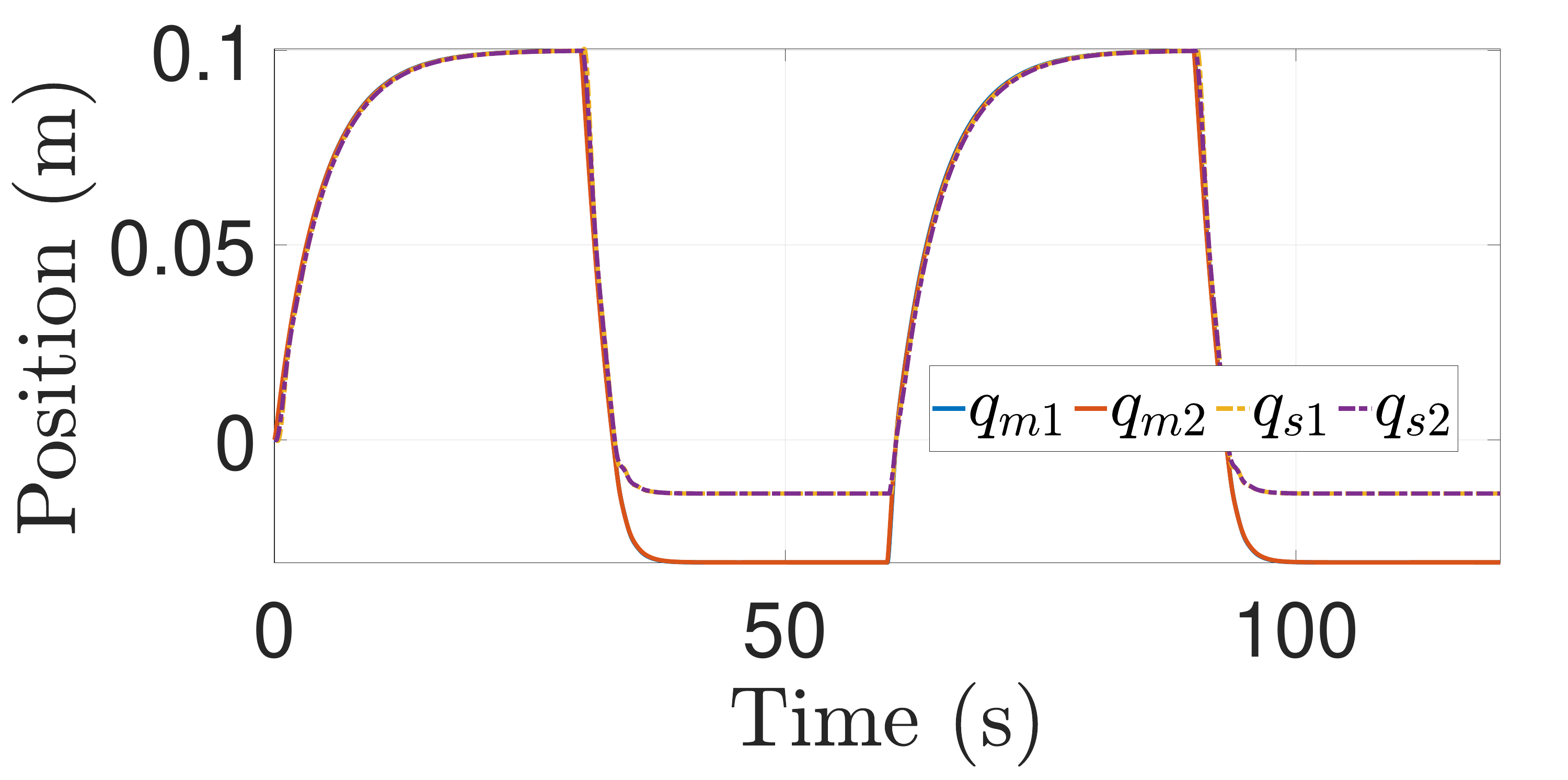}
\includegraphics[width=0.49\linewidth]{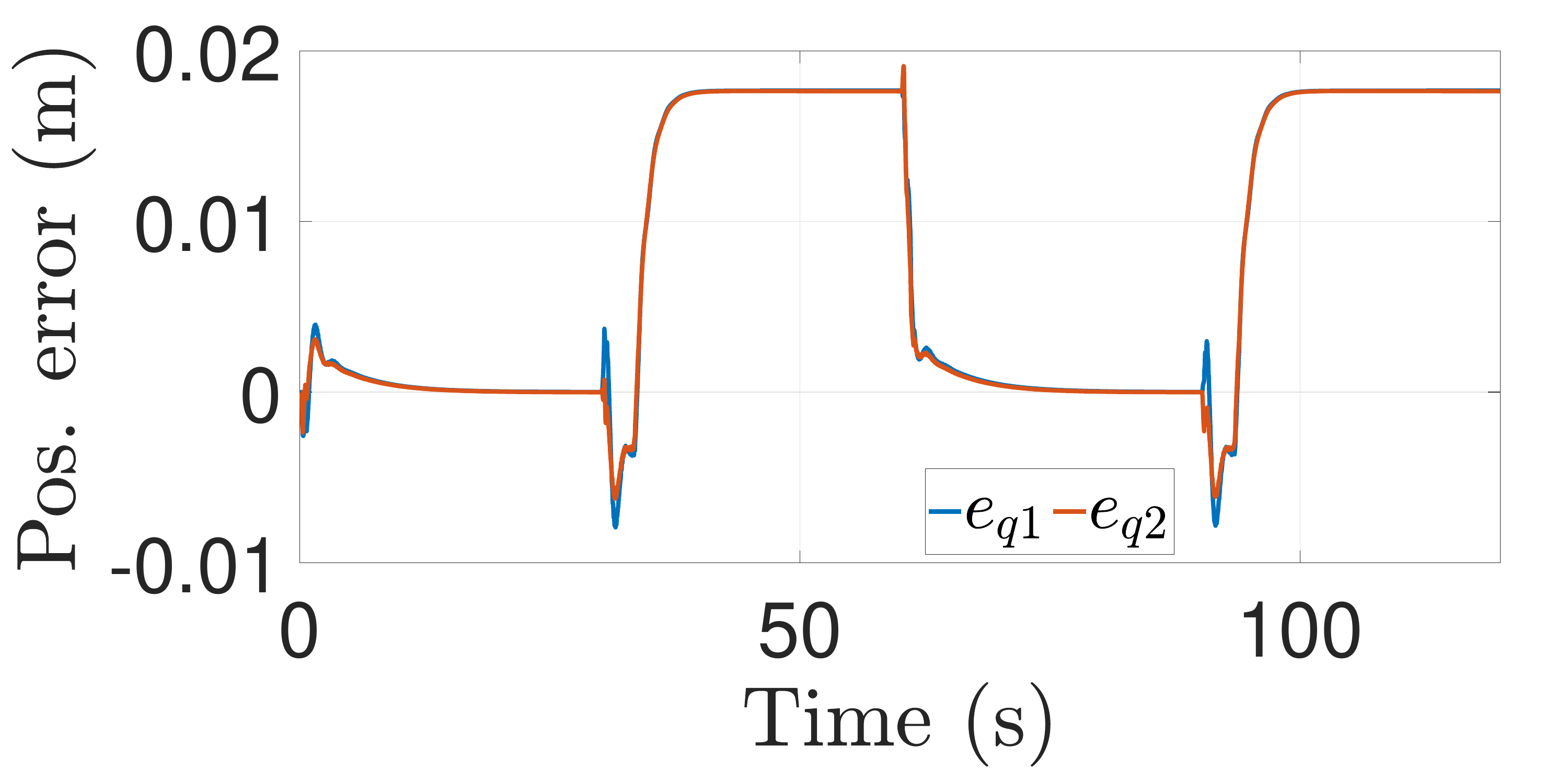}
    \caption{Stable position trajectories and tracking errors between the robots.}
    \label{fig:position}
\end{figure}

\begin{figure}[!t]
    \centering
\includegraphics[width=0.49\linewidth]{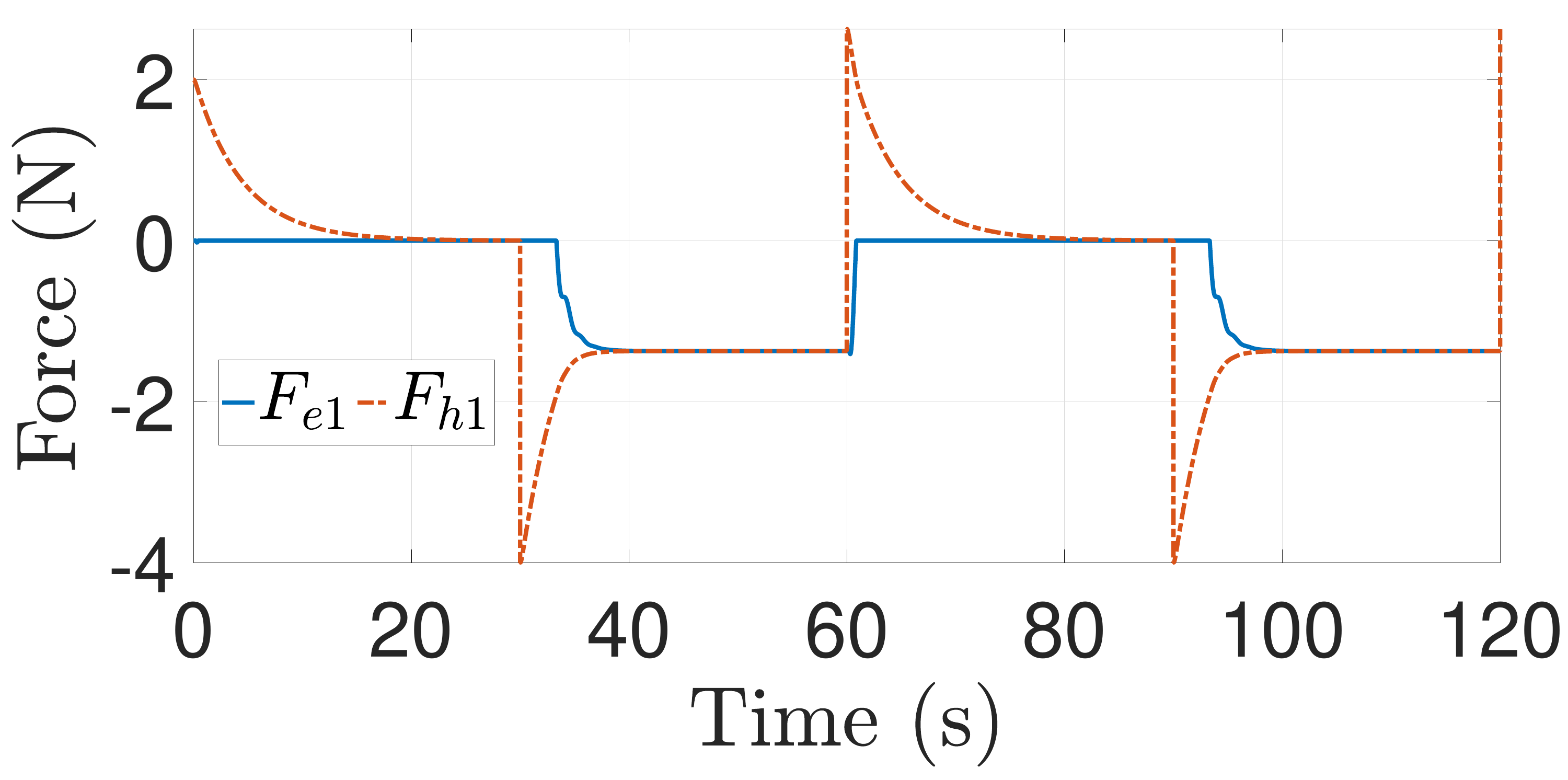}
\includegraphics[width=0.49\linewidth]{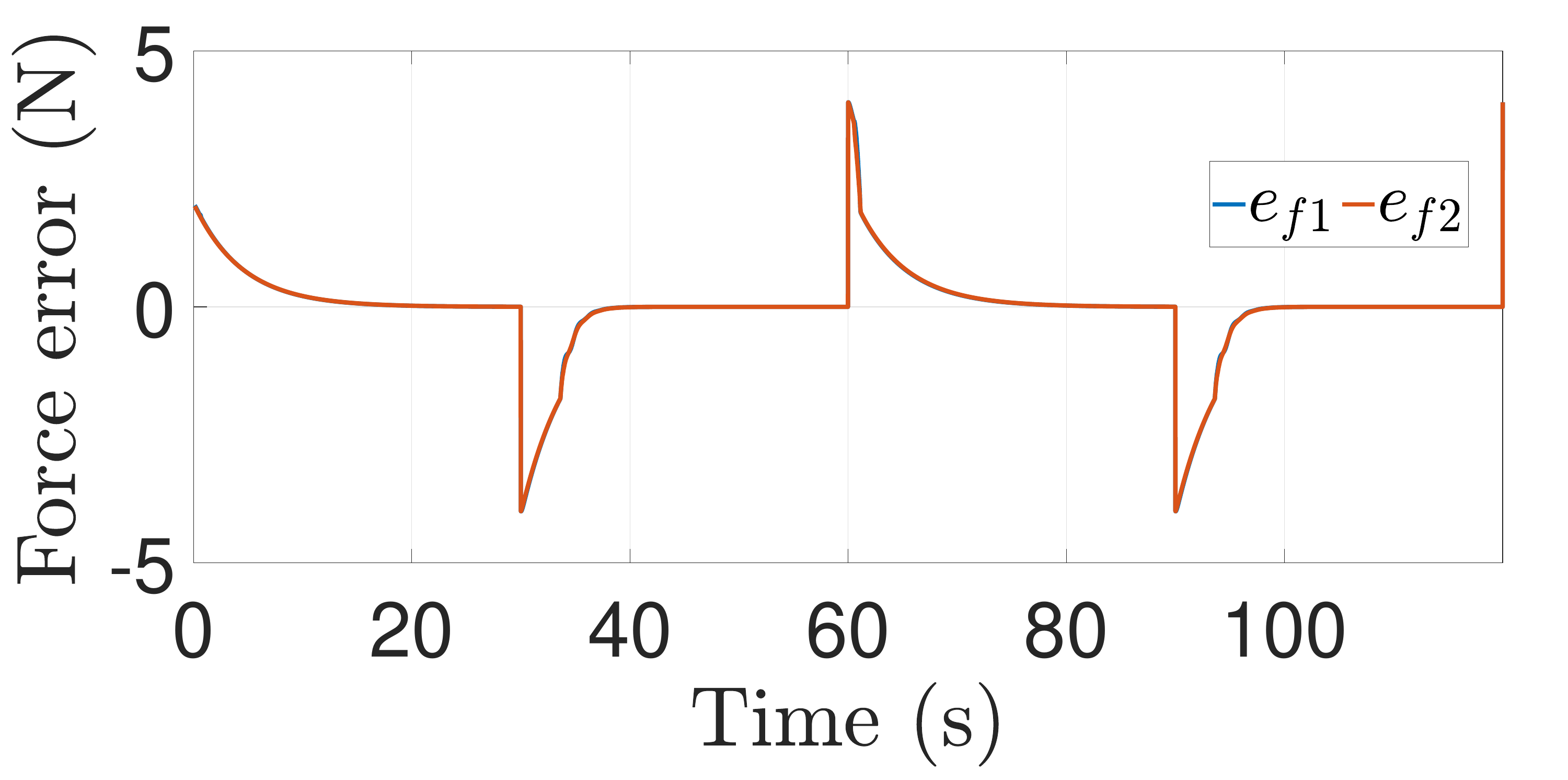}
    \caption{Force transmission results (first component only; the other one is similar) and force tracking errors.}
    \label{fig:force}
\end{figure}

Fig.~\ref{fig:force} presents effective force transmission in both cases, so the human can accurately feel the force exerted by the environment remotely. The root-mean-square error of force tracking over the simulation horizon is found to be $1.2993$~(N). Apart from \texttt{patternsearch}, we have experimented with the \texttt{particleswarm} and \texttt{surrogateopt} optimization schemes on MATLAB, giving similar or slightly worse results. Our simulation results illustrate transmission of positions in free motion and of forces in hard contact.

\section{CONCLUSION}

This paper addresses bilateral teleoperation between nonlinear robots
with constant communication delays while directly transmitting the
environmental contact force. A passivity-shortage characterization of
the nonlinear slave dynamics is derived, leading to an LMI
condition ensuring a dissipation inequality. The USP communication law is then applied to guarantee stability of
the delayed teleoperation loop and transparency in position and force transmission under appropriate conditions. The channel parameters
are tuned via optimization to improve 
transparency.

Future work will investigate the case of dimension mismatch between master and slave, extended neural network-inspired communication channels~\cite{tho2025neural} capable of
transmitting richer information such as velocities,
as well as experimental validation on robotic teleoperation platforms~\cite{pattanapong}.

\bibliographystyle{IEEEtran}
\bibliography{ref}
\end{document}